\documentclass[11pt, a4paper]{article}
\usepackage{amsmath, amsthm, amssymb, url}
\usepackage[margin=1in]{geometry}
\usepackage[english]{babel}
\usepackage{multirow}
\usepackage{authblk}
\usepackage{comment}
\usepackage{graphicx}
\usepackage{tikz}

\usepackage{type1cm}        
\usepackage{makeidx}         
\usepackage{graphicx}        
\usepackage{multicol}        
\usepackage[bottom]{footmisc}

\usepackage{newtxtext}       %
\usepackage[varvw]{newtxmath}       

\usepackage{multirow}
\usepackage{tikz}
\usepackage{tikz-cd}
\newcommand{\End}{\mathrm{End}}

\newcommand{\id}{\mathrm{id}}

\newcommand{\CA}{\mathrm{CA}}
\newcommand{\ICA}{\mathrm{ICA}}

\newcommand{\Aut}{\mathrm{Aut}}
\newcommand{\rev}{\mathrm{rev}}

\theoremstyle{plain}

\newtheorem{lemma}{Lemma}
\newtheorem{proposition}{Proposition}
\newtheorem{theorem}{Theorem}

\theoremstyle{definition}
\newtheorem{definition}{Definition}

\begin{document}

\title{A study on the composition of elementary cellular automata}
\author{Alonso Castillo-Ramirez\footnote{Email: alonso.castillor@academicos.udg.mx} \and Maria G. Maga\~{n}a-Chavez\footnote{Email: maria.magana3917@alumnos.udg.mx} \\
\small{Departamento de Matem\'aticas, Centro Universitario de Ciencias Exactas e Ingenier\'ias,\\ Universidad de Guadalajara, M\'exico.} }

\maketitle

\begin{abstract}
Elementary cellular automata (ECA) are one-dimensional discrete models of computation with a small memory set that have gained significant interest since the pioneer work of Stephen Wolfram, who studied them as time-discrete dynamical systems. Each of the 256 ECA is labeled as rule $X$, where $X$ is an integer between $0$ and $255$. An important property, that is usually overlooked in computational studies, is that the composition of any two one-dimensional cellular automata is again a one-dimensional cellular automaton. In this chapter, we begin a systematic study of the composition of ECA. Intuitively speaking, we shall consider that rule $X$ has low complexity if the compositions $X \circ Y$ and $Y \circ X$ have small minimal memory sets, for many rules $Y$. Hence, we propose a new classification of ECA based on the compositions among them. We also describe all \emph{semigroups} of ECA (i.e., composition-closed sets of ECA) and analyze their basic structure from the perspective of semigroup theory. In particular, we determine that the largest semigroups of ECA have $9$ elements, and have a subsemigroup of order $8$ that is \emph{$\mathcal{R}$-trivial}, property which has been recently used to define random walks and Markov chains over semigroups.  
 \\

\end{abstract}

\section{Introduction}\label{intro}

Elementary cellular automata (ECA) are one-dimensional cellular automata over a binary alphabet $A=\{0,1\}$ that admit a memory set $\{-1,0,1 \}$. Their systematic study was started by Stephen Wolfram in the 1980s, who labeled each ECA as rule $X$, with $X$ an integer between $0$ and $255$ (\cite{Wolfram}). Moreover, by considering them as  time-discrete dynamical systems, Wolfram classified the ECA into four classes of increasing order of complexity: uniform, periodic, chaotic, and complex behavior (\cite{Wolfram, Martinez}). A major breakthrough occurred when Matthew Cook proved that rule $110$ is capable of universal computation (\cite{Cook}). 

In this chapter, we focus on an operation between cellular automata that is often overlooked in computational research: the composition of their global functions. Since the composition of two one-dimensional cellular automata is again a one-dimensional cellular automaton, the set of all one-dimensional cellular automata, usually denoted by $\End( A^{\mathbb{Z}})$ or $\CA(\mathbb{Z};A)$, forms an algebraic structure known as a \emph{monoid} (i.e. a set equipped with a binary associative operation and an identity element). The group of invertible elements of this monoid, usually denoted by $\Aut(A^{\mathbb{Z}})$ or $\ICA(\mathbb{Z};A)$, is an important object of study in symbolic dynamics (e.g., see \cite{BLR88}, \cite[Sec. 13.2]{LM95}, \cite[Sec. 22]{B08}).

Our goal in this chapter is to start a computational exploration of the behavior of the composition between ECA. It is easy to realize that the composition of two ECA is not necessarily an ECA; however, how often does this happen? It turns out that, if we choose two rules $X$ and $Y$ at random, the probability that the composition $X \circ Y$ is again an ECA is approximately $6.29\%$. If we exclude rules $0$, $51$, $204$ and $255$ from our choices (since composition of these rules with any ECA always gives an ECA), the probability that $X \circ Y$ is an ECA reduces to $3.3\%$. Given a rule $X$, we say that another rule $Y$ is a left, or right, \emph{companion} of $X$, if $Y \circ X$, or $X \circ Y$, is an ECA, respectively. We observe that more complex rules, according to Wolfram, trend to have fewer left and right companions (Table \ref{average}). 

We propose a classification the complexity of the 256 ECA based on their behavior of the composition with other ECA. For this classification, the notion of left and right companions is too rigid, so we define \emph{quasi-elementary cellular automata} (QECA) as those one-dimensional cellular automata whose memory set is contained in $\{k-1, k, k+1 \}$, for some $k \in \mathbb{Z}$. We say that rule $Y$ is a left (resp. right) \emph{quasi-companion} of rule $X$ if $Y \circ X$ (resp. $X \circ Y$) is a QECA. Hence, we identify each rule $X$ with a pair of non-negative integers $(l_X, r_X)$, where $l_X$ is the number of left quasi-companions of $X$ and $r_X$ is the number of right quasi-companions of $X$. This identifies all 256 rules with points in the euclidean space $\mathbb{R}^2$, so we may apply a hierarchical clustering in order to obtain four classes (labeled from 0 to 3), ordered by decreasing number of quasi-companions (Table \ref{clasificacion_comp_cuasi}). Our classification provides a non-trivial way of grouping ECA that is not based on the size of their minimal memory set, and it is different from Wolfram's classification. 

In the last part of this chapter, we study ECA from a semigroup theory perspective. Recall that a \emph{semigroup} is a set equipped with an associative binary operation, and there is no significant difference between monoids and semigroups, as an identity element may always be added or subtracted to obtain one or the other. We begin by studying the \emph{natural partial order} (\cite[Sec. 1.8]{CP}) defined on \emph{idempotent} ECA (i.e. rules $X$ such that $X \circ X = X$). Next, we obtain and analyze all semigroups whose elements are all ECA. It turns out there are precisely seven \emph{ECA-maximal monoids} $M_i$, $i=1, \dots, 7$, in the sense that they are monoids of ECA that are not properly contained in any other monoid of ECA (Proposition \ref{monoids}). The symmetries of complements and reflection of one-dimensional cellular automata allow us to show that $M_1$ is isomorphic to $M_2$, $M_3$ is isomorphic to $M_4$, and $M_6$ is isomorphic to $M_7$. There are no further isomorphisms between these monoids. We provide the Cayley tables and the Green's equivalence classes of all these monoids. The most remarkable of these monoids is $M_1 \cong M_2$, as it has order $9$, and it contains a submonoid $M_1^\prime$ of order $8$ that is right reversible and $\mathcal{R}$-trivial, property which has attracted significant interest as it may be used to define random walks and Markov chains on monoids (see \cite{Markov}). 

The structure of this chapter is as follows. In Section \ref{basic}, we establish the notation and recall some basic results on cellular automata, including key properties of the well-known symmetries of reflection and complement. In Section \ref{sec-comp}, we study the compositions between ECA and define the notions of left and right companions and quasi-companions; this allow us to propose our classification of ECA according to their number of left and right quasi-companions. In Section \ref{sec-semigroup}, we review some basic concepts on semigroup theory, including the natural partial order on idempotents and Green's relations. In Section \ref{sec-semiECA}, we begin by describing the natural partial order on idempotent ECA. Then, we prove that any semigroup consisting of ECA must be contained in one of seven monoids $M_i$, $i=1, \dots, 7$. We finish this section by examining the structure of these seven monoids. Finally, in Section \ref{sec-con}, we discuss the main conclusions of the chapter and provide some possibilities for future work. In Appendix \ref{appendix}, we provide tables that contain, for each representative of the 88 equivalence classes of ECA, the number of left and right companions, number of ECA factorizations, and the number of left and right quasi-companions. The representatives of ECA are ordered in the tables according to their Wolfram's class. 

Many of the new results in this chapter are part of the MSc Thesis \cite{Magana} written by the second author at the University of Guadalajara.

\section{Basic theory of cellular automata}\label{basic}

Throughout this chapter, we shall consider the alphabet $A := \{ 0,1\}$. The \emph{full shift}, or \emph{configuration space}, $A^\mathbb{Z}$ is the set of all maps $x : \mathbb{Z} \to A$ equipped with the \emph{shift action} of $\mathbb{Z}$ on $A^\mathbb{Z}$ defined by 
\[ (k \cdot x) (s) := x(s - k) \]
for every $s,k \in \mathbb{Z}$, $x \in A^\mathbb{Z}$. An element $x$ of $A^{\mathbb{Z}}$ may be seen as a bi-infinite sequence 
\[  x =  \dots \; x(-2) \; x(-1) \; . \; x(0) \; x(1) \; x(2) \; \dots,  \]
where $x(s) \in A$, for all $s \in \mathbb{Z}$, and we write a dot before zeroth term of the sequence. With this notation, it may be readily seen that the shift action of $k \in \mathbb{Z}$ on $x$ represents the translation of the sequence $k$ places towards the right if $k$ is positive, and $-k$ places towards the left if $k$ is negative: 
\[ k \cdot x = \dots \; x(-2-k) \; x(-1-k) \; . \; x(-k) \; x(1-k) \; x(2-k) \; \dots.   \] 

The set $A^\mathbb{Z}$ may also be equipped with the \emph{prodiscrete topology}, which is the product topology of the discrete topology on $A$. Some basic facts of the prodiscrete topology on $A^\mathbb{Z}$ are described in \cite[Ch. 1]{CSC10}. 

\begin{definition}
A \emph{cellular automaton} over $A^\mathbb{Z}$ is a function $\tau : A^\mathbb{Z} \to A^\mathbb{Z}$ such that there is a finite subset $S \subseteq \mathbb{Z}$, called a \emph{memory set} of $\tau$, and a \emph{local function} $\mu : A^S \to A$ satisfying
\[ \tau(x)(s) = \mu (( (-s) \cdot x) \vert_{S}),  \quad \forall x \in A^\mathbb{Z}, s \in \mathbb{Z}.  \]
\end{definition}

Any cellular automaton $\tau : A^\mathbb{Z} \to A^\mathbb{Z}$ is continuous in the prodiscrete topology of $A^\mathbb{Z}$, and \emph{$\mathbb{Z}$-equivariant}, which means that 
\[ \tau(k \cdot x) = k \cdot \tau(x), \text{ for all } x \in A^\mathbb{Z}, k \in \mathbb{Z}. \]
Furthermore, the Curtis-Hedlund-Lyndon theorem (see \cite[Theorem 1.8.1]{CSC10}) guarantees that any continuous and $\mathbb{Z}$-equivariant self-function $A^\mathbb{Z}$ is a cellular automaton. 

A memory set for a cellular automaton $\tau$ is not unique: if $S$ is a memory set of $\tau$, then any finite superset $S^\prime \supseteq S$ is also a memory set of $\tau$. However, every cellular automaton has a unique \emph{minimal memory set} (MMS), i.e. a memory set of minimal cardinality, which is the intersection of all memory sets admitted by $\tau$ (see \cite[Prop. 1.5.2]{CSC10}). 

\begin{definition}
An \emph{elementary cellular automaton} (ECA) is a cellular automaton over $A^\mathbb{Z}$ whose minimal memory set is contained in $\{ -1, 0 ,1 \}$.
\end{definition} 

ECA are labeled as `rule $X$', where $X$ is a number from $0$ to $255$. In each case, the local rule $\mu_X : A^{\{ -1,0,1 \}} \to A$ of rule $X$ is determined as follows: let $X_1 \dots X_8$ be the binary representation of $X$ and write the elements of $A^{\{ -1,0,1 \}}$ in lexicographical descending order, i.e. $111, 110, \dots, 000$; then, the image of the $i$-th element of $A^{\{ -1,0,1 \}}$ under $\mu_X$ is $X_i$. We shall denote the global function of rule $X$ by the same integer $X$, i.e. $X : A^\mathbb{Z} \to A^\mathbb{Z}$.

For example, define $\mu : A^{ \{ -1,0,1 \} } \to A$ by the following table
\[ \begin{tabular}{c|cccccccc}
$x\in A^S$ & $111$ & $110$ & $101$ & $100$ & $011$ & $010$ & $001$ & $000$ \\ \hline
$\mu(x)$ & $0$ & $1$ & $1$ & $0$ & $1$ & $1$ & $1$ & $0$
\end{tabular}\] 
The ECA defined by the local rule $\mu : A^{ \{ -1,0,1 \} } \to A$ is rule 110, as the number in the second row of the above table is $110$ in binary representation.   

The following result is the main inspiration for our study (see \cite[1.4.9]{CSC10}). 

\begin{theorem}\label{composition}
Let $\sigma : A^\mathbb{Z} \to A^\mathbb{Z}$ and $\tau : A^\mathbb{Z} \to A^\mathbb{Z}$ be cellular automata with memory sets $S$ and $T$, respectively. Then, the composition $\sigma \circ \tau : A^\mathbb{Z} \to A^\mathbb{Z}$ is a cellular automaton with memory set $S+T := \{ s+t : s  \in S, t \in T \}$. 
\end{theorem}

The local rule of $\sigma \circ \tau$ may be described in terms of the local rules of $\sigma$ and $\tau$ as described in \cite[Remark 1.4.10]{CSC10}. It is important to note that Theorem \ref{composition} gives a memory set for $\sigma \circ \tau$ that is not necessarily a MMS; this implies that the MMS of $\sigma \circ \tau$ could be any subset of $S+T$. 

In \cite{Wolfram} (see also \cite[Sec. 2.1.]{Martinez}), Wolfram classified the ECA into four classes according to their behavior after being iterated many times from a random initial condition:
\begin{itemize}
\item \textbf{Class I}: uniform behavior.  
\item \textbf{Class II}: periodic behavior. 
\item \textbf{Class III}: chaotic behavior. 
\item \textbf{Class IV}: complex behavior. 
\end{itemize}

\begin{table}[!h]\centering
\caption{Wolfram's classification of ECA.}\label{class_wolfram}
\begin{tabular}{ll}\hline
\textbf{Wolfram's} &\textbf{Rules} \\ 
\textbf{classes} &  \\ \hline
Class I &0 8 32 40 128 136 160 168 \\
\hline
\multirow{5}{*}{Class II} & 1 2 3 4 5 6 7 9 10 11 12 13 14 15 19 23 24 \\
& 25 26 27 28 29 33 34 35 36 37 38 42 43 44 \\
& 46 50 51 56 57 58 62 72 73 74 76 77 78 94 \\
& 104 108 130 132 134 138 140 142 152 154 \\
& 156 162 164 170 172 178 184 200 204 232 \\
\hline
Class III & 18 22 30 45 60 90 105 122 126 146 150 \\
\hline
Class IV & 41 54 106 110 \\
\hline
\end{tabular}
\end{table}

In the next section we shall divide the $256$ ECA into $88$ equivalence classes that preserve many dynamic and algebraic properties of cellular automata (see Table \ref{equivalences}). Table \ref{class_wolfram} organizes ECA per Wolfram's class giving only one representative of each equivalence class of ECA.  


\subsection{Equivalences of cellular automata}\label{sec-equiv}

For any $x\in A^{\mathbb{Z}}$, denote by $x^{\rev}$ the reflection of $x$ through $0$; in other words, $x^{\textup{rev}}(k):=x(-k), \textup{ for all }k\in \mathbb{Z}$, which means that 
\[ x^\rev = \dots \; x(2) \; x(1) \; . \; x(0) \; x(-1) \; x(-2) \; \dots \]

Now, for any cellular automaton $\tau : A^\mathbb{Z} \to A^\mathbb{Z}$, define its \emph{reflection} $\tau^\star:A^\mathbb{Z}\rightarrow A^\mathbb{Z}$ by
\[ \tau^\star(x)=\left(\tau(x^\rev)\right)^\rev, \textup{ for all }x\in A^\mathbb{Z}. \]

The reflection of a cellular automaton satisfies the following properties (see \cite[Prop. 1]{CRG20} for the proof). 

\begin{lemma}\label{le-star}
Let $\tau : A^\mathbb{Z} \to A^\mathbb{Z}$ be a cellular automaton with memory set $T$.  
\begin{enumerate}
\item $\tau^\star : A^\mathbb{Z} \to A^\mathbb{Z}$ is a cellular automaton with memory set $-T := \{ -t : t \in T\}$.
\item $(\tau^\star)^\star = \tau$. 
\item For any other cellular automaton $\sigma : A^\mathbb{Z} \to A^\mathbb{Z}$, 
\[ (\sigma \circ \tau)^\star = \sigma^\star \circ \tau^\star.   \]
\end{enumerate}
\end{lemma}

In practice, it is easy to obtain the reflection of an ECA from its local rule table: we just have to apply the operation $\rev$ to the patterns in the first row of the table. For example, the local rule table of the reflection of rule 110 is
 \[ \begin{tabular}{c|cccccccc}
$x\in A^S$ & $111$ & $110$ & $101$ & $100$ & $011$ & $010$ & $001$ & $000$ \\ \hline
$\mu(x)$ & $0$ & $1$ & $1$ & $1$ & $1$ & $1$ & $0$ & $0$
\end{tabular}\] 
Hence, the reflection of rule 110 is rule 124. 

Another important symmetry between cellular automata is defined using rule $51$, which is the invertible ECA that exchanges $0$'s and $1$'s in a configuration in $A^{\mathbb{Z}}$. For any cellular automaton $\tau : A^\mathbb{Z} \to A^\mathbb{Z}$, define the \emph{complement} $\tau^c :  A^\mathbb{Z} \to A^\mathbb{Z}$ by
\[ \tau^c := 51 \circ \tau \circ 51. \]

Since $51 \circ 51 = \id$, it is easy to prove the following result. 
\begin{lemma}\label{le-c}
Let $\tau : A^\mathbb{Z} \to A^\mathbb{Z}$ be a cellular automaton with memory set $T$.
\begin{enumerate}
\item $\tau^c : A^\mathbb{Z} \to A^\mathbb{Z}$ is a cellular automaton with memory set $T$.
\item $(\tau^c)^c = \tau$. 
\item For any other cellular automaton $\sigma : A^\mathbb{Z} \to A^\mathbb{Z}$, 
\[ (\sigma \circ \tau)^c = \sigma^c \circ \tau^c.   \]
\end{enumerate}
\end{lemma}

The complement of an ECA may be obtained from its local rule table by exchanging $0$'s and $1$'s in both rows of the table. For example, the local rule table of the complement of rule 110 is
 \[ \begin{tabular}{c|cccccccc}
$x\in A^S$ & $111$ & $110$ & $101$ & $100$ & $011$ & $010$ & $001$ & $000$ \\ \hline
$\mu(x)$ & $1$ & $0$ & $0$ & $0$ & $1$ & $0$ & $0$ & $1$
\end{tabular}\] 
Hence, the complement of rule 110 is rule 137. 

\begin{lemma}
For any cellular automaton $\tau : A^\mathbb{Z} \to A^\mathbb{Z}$, we have 
\[ (\tau^c)^\star = (\tau^\star)^c.  \]
\end{lemma}
\begin{proof}
First observe that $(51)^\star = 51$. Hence, we have
\begin{align*}
 (\tau^c)^\star & = ( 51 \circ \tau \circ 51 )^\star  \\
 & =  (51 )^\star \circ \tau^\star \circ (51)^\star \\
 & = 51 \circ \tau^\star \circ 51 \\
 & = (\tau^\star)^c. 
\end{align*}
The result follows. 
\end{proof}

Reflections and complements allow us to define an equivalence class for a cellular automaton $\tau : A^\mathbb{Z} \to A^\mathbb{Z}$, which is the same as the orbit of $\tau$ under the action of the group generated by the reflection and complement symmetries:
\[ [\tau] := \{ \tau , \tau^\star, \tau^c , (\tau^c)^\star \}. \] 

\begin{table}[!htb]\centering
\caption{Equivalences of elementary cellular automata}\label{equivalences}
\scriptsize
\begin{tabular}{|c|c|ccc|c|c|ccc|}
\hline 
Rule & Size of MMS &\multicolumn{3}{c}{Equivalent rules} \vline & Rule &Size of MMS  &\multicolumn{3}{c}{Equivalent rules} \vline \\\hline
0 &0 &255 & & &56 &3 &98 &185 &227 \\
1 &3 &127 & & &57 &3 &99 & & \\
2 &3 &16 &191 &247 &58 &3 &114 &163 &177 \\
3 &2 &17 &63 &119 &60 &2 &102 &153 &195 \\
4 &3 &223 & & &62 &3 &118 &131 &145 \\
5 &2 &95 & & &72 &3 &237 & & \\
6 &3 &20 &159 &215 &73 &3 &109 & & \\
7 &3 &21 &31 &87 &74 &3 &88 &173 &229 \\
8 &3 &64 &239 &253 &76 &3 &205 & & \\
9 &3 &65 &111 &125 &77 &3 & & & \\
10 &2 &80 &175 &245 &78 &3 &92 &141 &197 \\
11 &3 &47 &81 &117 &90 &2 &165 & & \\
12 &2 &68 &207 &221 &94 &3 &133 & & \\
13 &3 &69 &79 &93 &104 &3 &233 & & \\
14 &3 &84 &143 &213 &105 &3 & & & \\
15 &1 &85 & & &106 &3 &120 &169 &225 \\
18 &3 &183 & & &108 &3 &201 & & \\
19 &3 &55 & & &110 &3 &124 &137 &193 \\
22 &3 &151 & & &122 &3 &161 & & \\
23 &3 & & & &126 &3 &129 & & \\
24 &3 &66 &189 &231 &128 &3 &254 & & \\
25 &3 &61 &67 &103 &130 &3 &144 &190 &246 \\
26 &3 &82 &167 &181 &132 &3 &222 & & \\
27 &3 &39 &53 &83 &134 &3 &148 &158 &214 \\
28 &3 &70 &157 &199 &136 &2 &192 &238 &252 \\
29 &3 &71 & & &138 &3 &174 &208 &244 \\
30 &3 &86 &135 &149 &140 &3 &196 &206 &220 \\
32 &3 &251 & & &142 &3 &212 & & \\
33 &3 &123 & & &146 &3 &182 & & \\
34 &2 &48 &187 &243 &150 &3 & & & \\
35 &3 &49 &59 &115 &152 &3 &188 &194 &230 \\
36 &3 &219 & & &154 &3 &166 &180 &210 \\
37 &3 &91 & & &156 &3 &198 & & \\
38 &3 &52 &155 &211 &160 &2 &250 & & \\
40 &3 &96 &235 &249 &162 &3 &176 &186 &242 \\
41 &3 &97 &107 &121 &164 &3 &218 & & \\
42 &3 &112 &171 &241 &168 &3 &224 &234 &248 \\
43 &3 &113 & & &170 &1 &240 & & \\
44 &3 &100 &203 &217 &172 &3 &202 &216 &228 \\
45 &3 &75 &89 &101 &178 &3 & & & \\
46 &3 &116 &139 &209 &184 &3 &226 & & \\
50 &3 &179 & & &200 &3 &236 & & \\
51 &1 & & & &204 &1 & & & \\
54 &3 &147 & & &232 &3 & & & \\
\hline
\end{tabular}
\end{table}

This divides the 256 ECA into 88 equivalence classes as described in \cite[Table 1]{Martinez} and \cite[Table 1]{Wolfram}. Many dynamical and algebraic properties of ECA are preserved by these equivalence classes. Table \ref{equivalences} shows this 88 equivalence classes and the size of their minimal memory set. 


\section{Compositions of elementary cellular automata}\label{sec-comp}

In this section, we study how composition behaves on the set of ECA. Given two ECA, $X$ and $Y$, it is easy to realize that $X \circ Y$ is not necessarily an ECA; by Theorem \ref{composition}, a memory set for $X \circ Y$ is 
\[ \{-1,0,1 \} + \{-1,0,1 \} = \{ -2, -1, 0 , 1, 2 \}, \]
but $X \circ Y$ will be an ECA if and only if its MMS is contained in $\{ -1,0,1\}$. There are no general theoretical results that allow us to determine the MMS of $X \circ Y$. 

We implemented in Python \cite{Python} a function to calculate the composition of any two ECA and a function to determine the MMS of this composition. 

\begin{figure}[!h]
\centering
\includegraphics[scale=0.45]{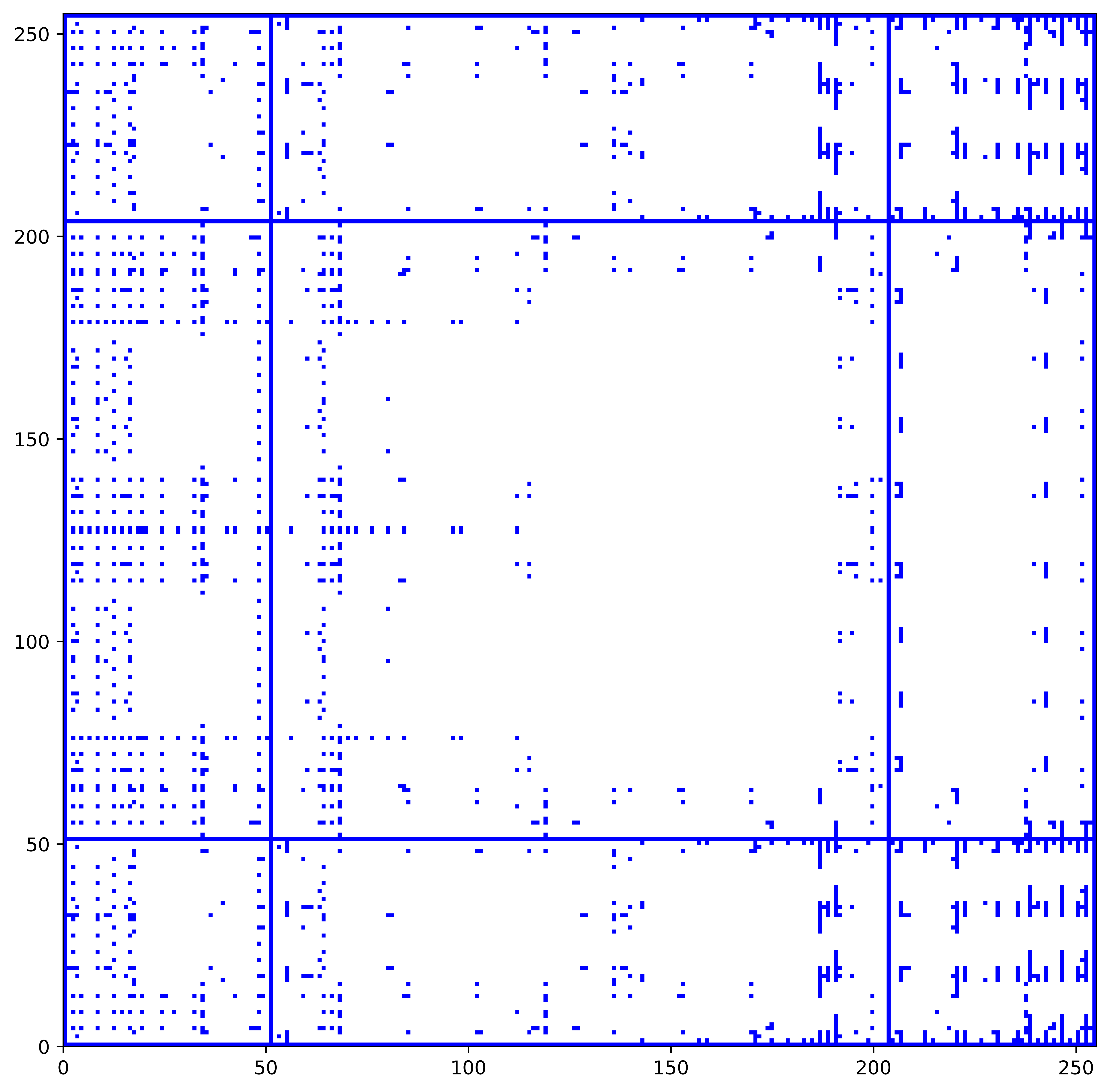}
\caption{Compositions of ECA giving ECA.}
\label{fig-elemental}
\end{figure}

Fig. \ref{fig-elemental} is a graphical representation of the compositions of all ECA between themselves, where there is a blue dot at position $(X,Y)$ if and only if $X \circ Y$ is an ECA. This depicts few elementary results, such that the compositions $X \circ Y$ and $Y \circ X$ are always elementary for every $X \in \{ 0,51,204, 255\}$ and $0 \leq Y \leq 255$; this may have been easily deduced in advance because rules $0$ and $255$ are constant functions with empty MMS, while rules $51$ and $204$ have MMS equal to $\{0\}$ (in fact, rule $204$ is the identity function). 

Another interesting deduction we may obtain is that the probability that the composition of two ECA chosen at random is again an ECA is $\frac{4128}{65536} \approx 0.0629$. Furthermore, if we exclude rules $0$, $51$, $204$, and $255$, from our possible choices, then this probability drops to $\frac{2096}{63504} \approx 0.033$.

\begin{definition}
Let $X : A^\mathbb{Z} \to A^\mathbb{Z}$ be an ECA. 
\begin{enumerate}
\item Say that rule $Y$ is a \emph{left companion} of rule $X$ if the composition $Y \circ X$ is an $ECA$. 
\item Say that rule $Y$ is a \emph{right companion} of rule $X$ if the composition $X \circ Y$ is an $ECA$. 
\item A pair of rules $(Y, Z)$ is an \emph{ECA factorization} of rule $R_X$ if $X = Y \circ Z$. 
\end{enumerate}
\end{definition}

\begin{lemma}\label{le-preserve}
Let  $X$ and $Y$ be two equivalent rules, according to Table \ref{equivalences}. Then $X$ and $Y$ have the same number of left and right companions, and the same number of ECA factorizations. 
\end{lemma}
\begin{proof}
As $X$ and $Y$ are equivalent, there must exist $r \in \{ 1, \star, c, c \star \}$ such that $Y = X^r$. This means that $r$ is the identity symmetry, or the reflection, or the complement, or the complement reflection, discussed in Section \ref{sec-equiv}. 

Let $\mathcal{A}$ and $\mathcal{B}$ be sets of all left companions of rules $X$ and $Y$, respectively. Consider the function $r : \mathcal{A} \to \mathcal{B}$ as given by $Z \mapsto  Z^r$, for every $Z \in \mathcal{A}$ (formally, this is the restriction to $\mathcal{A}$ of the symmetry $r$). To show that this function is well-defined, take $Z \in \mathcal{A}$. Then, there exists a rule $W$ such that $X = Z \circ W$. By Lemmas \ref{le-star} and \ref{le-c}, we have 
\[ Y = X^r = (Z \circ W)^r  = Z^r \circ W^r.\] 
This shows that $Z^r \in \mathcal{B}$. Moreover, the function is injective because Lemmas \ref{le-star} (2.) and \ref{le-c} (2.) show that
\[ Z^r = W^r  \ \Rightarrow \ (Z^r)^r= (W^r)^r \ \Rightarrow \ Z = W.   \]  
The function $r : \mathcal{A} \to \mathcal{B}$ is surjective because a preimage of $L \in \mathcal{B}$ is $L^r \in \mathcal{A}$. Hence, we have a bijection between $\mathcal{A}$ and $\mathcal{B}$, so $\vert \mathcal{A} \vert = \vert \mathcal{B} \vert$. An analogous argument shows that $X$ and $Y$ have the same number of right companions.

In order to show that $X$ and $Y$ have the same number of ECA factorizations, we may use $r \in  \{ 1, \star, c, c \star \}$ in order to define a bijection between the sets
\[ \{ (Z, W) : X = Z \circ W \} \text{ and } \{ (L, S) : Y = L \circ S \}.   \]
The result follows by similar arguments as in the previous paragraph.
\end{proof}

Tables \ref{comp-II} and \ref{comp-I} in the Appendix give the number of left and right companions, and number of factorizations, for each representative of ECA equivalence class, according to Table \ref{equivalences}. In these tables, we have arranged the ECA according to their Wolfram's class, to show that, in general, the numbers of left and right companions and number of factorizations decrease as the number of Wolfram's class increase. 

\begin{table}[!h]\centering
\caption{Average of companions by Wolfram's class}\label{average}
\renewcommand{\arraystretch}{1.5}
\begin{tabular}{|l|r|r|r|r|}
\hline
\textbf{Wolfram's class} & \textbf{Average of left comp.} &\textbf{Average of right comp.} & \textbf{Average of factorizations}  \\\hline
\textbf{Class I} &18.18 &19.64 & 16.00 \\ \hline
\textbf{Class II} & 8.25 & 8.42 & 6.57 \\\hline
\textbf{Class III} & 3.85 & 0.31 & 3.23 \\\hline
\textbf{Class IV} &2.00 &0.00 &  0.00  \\
\hline
\end{tabular}
\end{table}

Table \ref{average} shows the average number of left and right companions and the average number of factorizations for each Wolfram's class. Rules $0$, $51$, $204$ and $255$ have been excluded from these averages, as they are extreme exceptions. We see a clear trend that when the number of Wolfram's class increases, all the averages decrease; this implies that the dynamical complexity of an ECA affects its behavior when composed with other ECA.     

The notion of left and right companions is too rigid, as it requires that the MMS of the composition is neither moved nor translated: it must stay inside $\{ -1, 0, 1 \}$. This rigidness is shown, for example, in the fact that the low-complexity rule $170$, which simply represents the left translation, has only $12$ left and right companions. In order to address this, we introduce the notion of \emph{quasi-elementary cellular automaton} and \emph{quasi-companion}.  

\begin{definition}
A \emph{quasi-elementary cellular automaton} (QECA) is a cellular automaton $\tau : A^{\mathbb{Z}} \to A^{\mathbb{Z}}$ whose minimal memory set is contained in $\{ k - 1, k, k+1  \}$ for some $k \in \mathbb{Z}$. 
\end{definition}

\begin{definition}
Let $X : A^\mathbb{Z} \to A^\mathbb{Z}$ be an ECA. 
\begin{enumerate}
\item A rule $Y$ is a \emph{left quasi-companion} of rule $X$ if the composition $Y \circ X$ is a $QECA$. 
\item A rule $Y$ is a \emph{right quasi-companion} of rule $X$ if the composition $X \circ Y$ is a $QECA$. 
\end{enumerate}
\end{definition}

\begin{figure}[h!]
\centering
\includegraphics[scale=0.5]{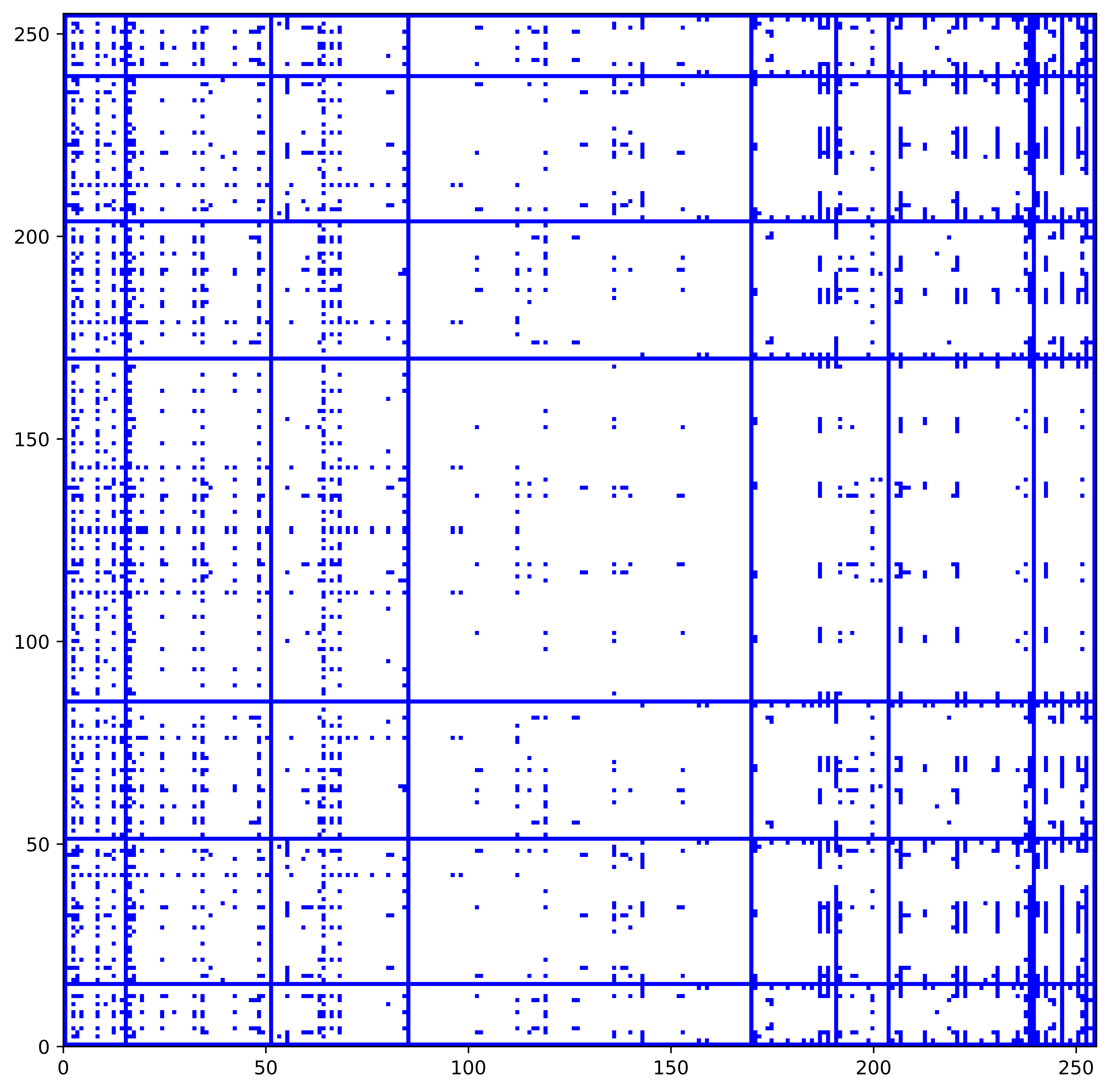}
\caption{Compositions of ECA giving QECA.}
\label{fig-qelemental}
\end{figure}

We used Python \cite{Python} to determine when the composition of two ECA gives a QECA. Figure \ref{fig-qelemental} is a graphical representation of the compositions of all ECA between themselves, where there is a blue dot at position $(X,Y)$ if and only if $X \circ Y$ is a QECA. 

Note that the composition of rules $0$, $15$, $51$, $85$, $170$, $204$, $240$, and $255$, with any ECA, on the right or on the left, always gives us a QECA.   

As in Lemma \ref{le-preserve}, we may show that any two equivalent ECA must have the same number of left and right quasi-companions. Tables \ref{comp-II} and \ref{comp-I} in the Appendix give the number of left and right quasi-companions for each representative of ECA equivalence class, according to Table \ref{equivalences}. 

Table \ref{average-quasi} shows the average number of left and right quasi-companions for each Wolfram's class. Rules $0$, $15$, $51$, $85$, $170$, $204$, $240$, and $255$, have been excluded from these averages, as they are extreme exceptions. Again we see an inversely proportional trend between the number of Wolfram's class and the number of left and right quasi-companions. 

\begin{table}[!h]\centering
\caption{Average of quasi-companions by Wolfram's class}\label{average-quasi}
\renewcommand{\arraystretch}{1.5}
\begin{tabular}{|l|r|r|r|}
\hline
\textbf{Wolfram's class} & \textbf{Average of left quasi-comp.} &\textbf{Average of right quasi-comp.}  \\\hline
\textbf{Class I} & 24.73 & 37.82 \\ \hline
\textbf{Class II} & 15.53 & 14.80 \\  \hline
\textbf{Class III} & 7.23 & 4.31\\ \hline
\textbf{Class IV} & 5.43 &0.00  \\
 \hline
\end{tabular}
\end{table}

Our next goal is to propose a new classification of ECA according to their left and right quasi-companions. In order to do this, to each rule $X$ we associate a pair of non-negative integers $( l_X, r_X)$, where $l_X$ is the number of left quasi-companions of $X$ and $r_X$ is the number of right quasi-companions of $X$. 

Figure \ref{fig-classification-cuasi} is the graph of the pairs $( l_x, r_x)$ for each representative of ECA equivalence class $X$.  

 \begin{figure}[!h]
\centering
\includegraphics[scale=0.8]{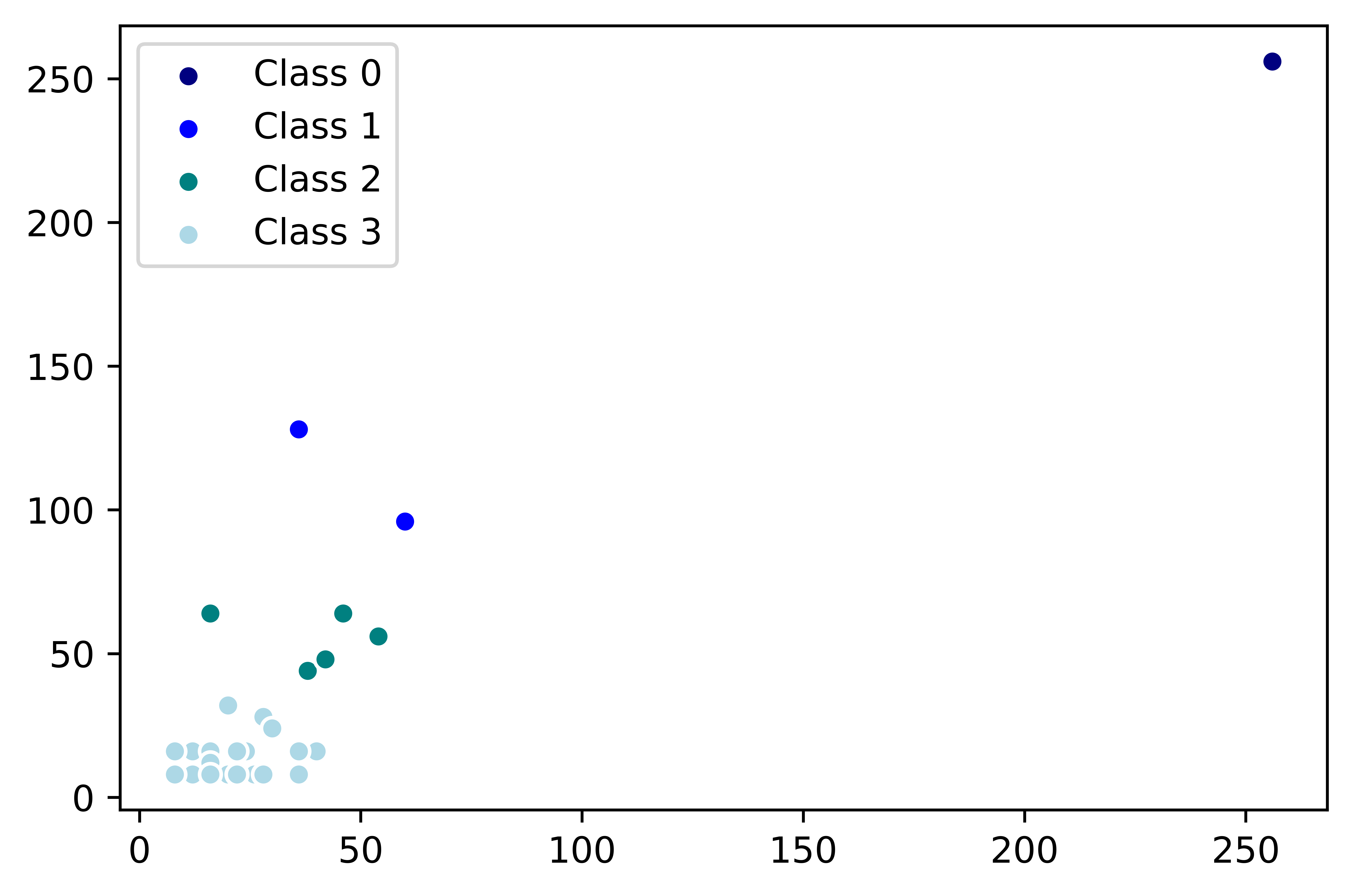}
\label{fig-classification-cuasi}
\caption{Graph of ECA represented by left and right quasi-companions.}
\end{figure}

We used the package scipy.cluster.hierarchy in Python \cite{Python} in order to do a hierarchical clustering of these data using the Euclidean distance of $\mathbb{R}^2$ as the metric. The \emph{silhouette coefficient} (see \cite[p. 87]{Kaufman}) when the data is grouped into $3$, $4$, $5$, and $6$ clusters is $0.7559$, $0.7281$, $0.5979$, and $0.6142$, respectively. This coefficient is a measure of how tightly grouped are the data in the clusters: the closer it is to $1$, the better the data have been classified. We decided to classify our data into $4$ clusters because it balances a good silhouette coefficient with a decent number of clusters.    

Table \ref{clasificacion_comp_cuasi} divides the $88$ classes of ECA into four classes (the four clusters obtained), which are also depicted by colors in Figure \ref{fig-classification-cuasi}. Class $0$ contains the ECA with the highest number of left and right quasi-companions, while Class $3$ contains the ECA with the lowest number of left and right quasi-companions.

\begin{table}[!h]\centering
\caption{Classification of ECA with respect to their number of left and right quasi-companions}\label{clasificacion_comp_cuasi}
\begin{tabular}{|c|c|l|}\hline
\textbf{Class}  & \textbf{No. of Rules} & \textbf{Rules} \\\hline
Class 0 &5 &0 15 51 170 204 \\
\hline 	
Class 1 &4 &2 8 12 34 \\
\hline
Class 2 &9 &3 4 14 19 24 32 42 136 200 \\
\hline
\multirow{5}{*}{Class 3} &\multirow{5}{*}{70} &1 5 6 7 9 10 11 13 18 22 23 25 26 27 28 29 30 33 \\
& &35 36 37 38 40 41 43 44 45 46 50 54 56 57 58 60 \\
& &62 72 73 74 76 77 78 90 94 104 105 106 108 110 \\
& &122 126 128 130 132 134 138 140 142 146 150 152 \\
& &154 156 160 162 164 168 172 178 184 232 \\
\hline
\end{tabular}
\end{table}

It is worth noting that rules with MMS of size $2$ and $3$ are present in Classes 1, 2 and 3. For example, in Class 1, rules $12$ and $34$ have MMS of size $2$, while rules $2$ and $8$ have MMS of size $3$. In Class 2, the rules $3$ and $136$ have MMS of size $2$, while the rest of the rules have MMS of size $3$. In Class 3, rules $5$, $10$, $60$, $90$, and $160$, have MMS of size $2$, while the rest of the rules have MMS of size $3$. Hence, our classification provides a non-trivial way of grouping ECA that is not based on the size of their MMS. 

Finally, observe that Wolfram's classes III and IV are contained in our Class 3. However, elements of Wolfram's classes I and II are present in all our four classes.  

\section{Basic semigroup theory}\label{sec-semigroup}

In this section, we review some basic results on semigroup theory. For a thorough introduction to this topic, see \cite{CP}.

\begin{definition}
A \emph{semigroup} is a set $S$ together with a binary operation $\circ : S \times S \to S$ that is associative; this means that
\[ (x \circ y) \circ z = x \circ ( y \circ z), \quad \forall x,y,z \in S. \]
A \emph{monoid} is a semigroup $M$ that has an identity element $I \in M$; this means that, 
\[ x \circ I = I \circ x = x, \quad \forall x \in M. \]
\end{definition}

A \emph{subsemigroup} of a semigroup $S$ is a subset of $S$ that is itself a semigroup when restricting the binary operation $\circ$. We define analogously a \emph{submonoid} of a monoid $M$. 

The order of a semigroup $S$ is simply the cardinality of the set $S$. The \emph{Cayley table} of a finite semigroup $S$ of order $n$ is a $n \times n$ table whose rows and columns are indexed by the elements $s_i$ of $S$ in a fixed order, and the $(i,j)$ entry of the table is the product $s_i \circ s_j$. 

Given a semigroup $S$ without identity element, it is easy to construct a monoid by defining $S^I : = S \cup \{ I\}$, where $I$ is a formal symbol, and extending the binary operation $\circ$ to $S^I$ by $I \circ x := x$ and $x \circ I := x$ for all $x \in S$. Hence, the theory of semigroups and the theory of monoids is essentially the same.

An element $x$ of a monoid $M$ is \emph{invertible} or a \emph{unit} if there exists $y \in M$ such that
\[ x \circ y = y \circ x = I.  \]
The inverse of a unit $x \in M$ is unique and denoted by $x^{-1}$. The \emph{group of units} of $M$, denoted by $U(M)$, is the set of all invertible elements of $M$.

\begin{definition}
We define special kinds of elements in a semigroup $S$. 
\begin{enumerate}
\item An element $O \in S$ is a \emph{zero} if
\[ O \circ x = x \circ O = O, \quad \forall x \in S. \]
\item An element $x \in S$ is an \emph{idempotent} if
\[ x^2 := x \circ x = x. \] 
The semigroup $S$ is called a \emph{band} if all of its elements are idempotents. 
\item If $S$ has a zero $O$, an element $x \in S$ is \emph{nilpotent} if there exists $n \in \mathbb{N}$ such that 
\[ x^n := \underbrace{x \circ x \circ \dots \circ x}_{n \text{ times }}= O. \] 
\end{enumerate}
\end{definition}

There is a partial order $\leq$ that may be given to any set of idempotents $E$ in a semigroup $S$ that is called the \emph{natural partial order} of $E$ (see \cite[Sec. 1.8]{CP}). This is defined as follows: for any $x, y \in E$, we say that $x \leq y$ if and only if $xy = yx = x$. For completeness, we shall prove here that this natural partial order is indeed a partial order.

\begin{proposition}
The natural partial order $\leq$ defined above on a set of idempotents $E$ in a semigroup $S$ is indeed a partial order. 
\end{proposition}
\begin{proof}
We shall prove that $\leq$ is a reflexive, antisymmetric and transitive relation.
\begin{description}
\item[Reflexive:] For any $x \in E$, we have that $x \circ x = x \circ x = x$, as $x$ is idempotent. Hence, $x \leq x$.
\item[Antisymmetric:] Suppose that $x \leq y$ and $y \leq x$, for some $x,y \in E$. By definition, we have
\[ x\circ y = y \circ x = x \quad \text{and} \quad y \circ x=x \circ y = y. \]
Therefore, $x=y$. 
\item[Transitive:] Suppose that $x \leq y$ and $y \leq z$, for some $x,y, z \in E$. By definition, we have
\[ x \circ y = y \circ x = x \quad \text{and} \quad y \circ z=z\circ y = y. \]
Operating by $x$ on the left in the second equality, we obtain
\[ x\circ (y \circ z) = x \circ y \ \Rightarrow \ x \circ z = x. \]
Now, operating by $x$ on the right in the second equality, we obtain
\[ (z \circ y) \circ x = y \circ x \ \Rightarrow \ z\circ x = x. \]
Therefore, $x \leq z$. 
\end{description}
\end{proof}

\begin{definition}
The \emph{left principal ideal}, \emph{right principal ideal}, and \emph{two-sided principal ideal} of $x \in S$ are defined, respectively,
\begin{align*}
S^{I}x & := \{ a \circ x : a \in S^{I} \}, \\
xS^{I} & := \{ x \circ a : a \in S^{I}  \}, \\
S^{I}xS^{I} & := \{ a \circ x \circ b : a,b \in S^{I}  \}.
\end{align*}
\end{definition}

\begin{definition}
The \emph{Green's relations} on a semigroup $S$ are the equivalence relations on $S$ defined as follows:
\begin{enumerate}
\item $x \mathcal{L} y$ if and only if $S^{I}x = S^{I}y$. 
\item $x \mathcal{R} y$ if and only if $xS^{I} = yS^{I}$.
\item $x \mathcal{J}  y$ if and only if $S^{I}xS^{I} = S^{I}yS^{I}$.
\item $x \mathcal{H} y$ if and only if $x \mathcal{L} y$ and $x \mathcal{R} y$.
\item $x \mathcal{D} y$ if and only if there exists $z \in M$ such that $x \mathcal{L} z$ and $z \mathcal{R} y$.
\end{enumerate}
\end{definition}

The Green's relations are fundamental in semigroup theory, as they are helpful in the analysis of the structure of a semigroup (see \cite[Ch. 2]{CP}). The $\mathcal{L}$ and $\mathcal{R}$ may be equivalently defined as the strongly connected components of the left and right Cayley graphs of $S^I$. 

A monoid $M$ is called \emph{$\mathcal{R}$-trivial} if the $\mathcal{R}$ equivalence relation is trivial (i.e. all equivalence classes have size $1$). This automatically implies that the $\mathcal{H}$ equivalence relation is also trivial, and that $\mathcal{L} = \mathcal{D}$. It turns out that an $\mathcal{R}$-trivial monoid $M$ provides a suitable setting to define Markov chains arising from random walks on $M$ \cite{Markov}. Although Markov chains have been recently studied over an arbitrary finite semigroup \cite{Rhodes}, special features of the representation theory of $\mathcal{R}$-trivial monoids make more amicable the analysis of the corresponding Markov chains.     

An \emph{isomorphism} between two semigroups $S_1$ and $S_2$ is a bijective function $f : S_1 \to S_2$ such that
\[ f(x \circ y) = f(x) \cdot f(y), \quad \forall x,y \in S_1, \text{ and } , \] 
where $\circ$ is the binary operation of $S_1$ and $\cdot$ is the binary operation of $S_2$. If such isomorphism exists, we say that $S_1$ and $S_2$ are \emph{isomorphic} and write $S_1 \cong S_2$. When $S_1 = S_2$, we call $f$ an \emph{automorphism}. In order to consider an isomorphism between two monoids $M_1$ and $M_2$, we add the condition $f(I_1) = I_2$, where $I_i$ is the identity element of $M_i$, for $i \in \{ 1,2 \}$.  Similarly, we define an \emph{anti-isomorphism} between $S_1$ and $S_2$ as a bijective function $f : S_1 \to S_2$ such that
\[ f(x \circ y) = f(y) \cdot f(x), \quad \forall x,y \in S_1. \]


\section{Semigroups of ECA}\label{sec-semiECA}

By Theorem \ref{composition}, the set of all cellular automata over $A^\mathbb{Z}$, usually denoted by $\End(A^\mathbb{Z})$, or by $\CA(\mathbb{Z}, A)$, is a monoid when equipped with the composition of functions, as this always satisfy the associative property. The identity element of $\End(A^{\mathbb{Z}})$ coincides with rule $204$. The group of units of $\End(A^\mathbb{Z})$ is usually denoted by $\Aut(A^\mathbb{Z})$ or $\ICA(\mathbb{Z}, A)$, and it is known to satisfy many exceptional properties. For example, it is known that $\Aut(A^\mathbb{Z})$ is countable, but not finitely generated, and contains a subgroup isomorphic to any finite group (see \cite[Sec. 13.2]{LM95}). By Lemmas \ref{le-star} and \ref{le-c}, the symmetries reflection and complement define automorphisms of the monoid $\End(A^{\mathbb{Z}})$. For an automorphism $r$ of $\End(A^{\mathbb{Z}})$, we shall use the notation $\tau^r$ to represent the image of $\tau \in \End(A^\mathbb{Z})$ under $r$.   

Now we shall study subsemigroups of $\End(A^\mathbb{Z})$ consisting entirely of elementary cellular automata. We begin by studying the $13$ ECA that are idempotent, since each one of them generates a subsemigroup with exactly one element. Explicitly, the idempotent ECA are the following:
\[ E := \{ 0, 4, 12, 68, 76, 200, 236, 204, 205, 207, 221, 223, 255 \}.\]
Figure \ref{hasse_idem} depicts the Hasse diagram of the natural order of idempotents on the set $E$ (recall that the \emph{Hasse diagram} is the transitive reduction of the graph of a finite partially ordered set).    

\begin{figure}[h]
\begin{center}
\begin{tikzpicture}
    \matrix (A) [matrix of nodes, row sep=1.5cm]
    { 
       ~ &  ~ & ~ &  ~ & ~ & $204$ \\  
	$4$ & $12$  & $68$ & $76$ & $200$ & $236$ & $205$ & $207$  & $221$ & $223$ \\
	~ & ~ & ~ &  $0$ &~ & ~ & ~ & $255$\\
    };
		\draw (A-1-6.south)--(A-2-1.north);
		\draw (A-1-6.south)--(A-2-2.north);
		\draw (A-1-6.south)--(A-2-3.north);
		\draw (A-1-6.south)--(A-2-4.north);
		\draw (A-1-6.south)--(A-2-5.north);
		\draw (A-1-6.south)--(A-2-6.north);
		\draw (A-1-6.south)--(A-2-7.north);
		\draw (A-1-6.south)--(A-2-8.north);
		\draw (A-1-6.south)--(A-2-9.north);
		\draw (A-1-6.south)--(A-2-10.north);
		
		\draw (A-2-1.south)--(A-3-4.north);
		\draw (A-2-2.south)--(A-3-4.north);
		\draw (A-2-3.south)--(A-3-4.north);
		\draw (A-2-4.south)--(A-3-4.north);
    \draw (A-2-5.south)--(A-3-4.north);
		\draw (A-2-6.south)--(A-3-4.north);
		
		\draw (A-2-5.south)--(A-3-8.north);
		\draw (A-2-6.south)--(A-3-8.north);
		\draw (A-2-7.south)--(A-3-8.north);
		\draw (A-2-8.south)--(A-3-8.north);
		\draw (A-2-9.south)--(A-3-8.north);
		\draw (A-2-10.south)--(A-3-8.north);
\end{tikzpicture}
\end{center}
\caption{Hasse diagram of the natural order of the idempotent ECA}
\label{hasse_idem}
\end{figure}
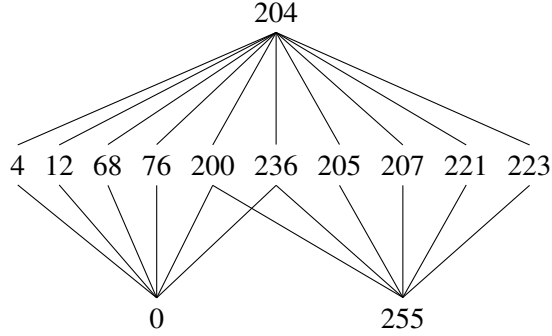

As expected, the identity function, rule $204$, is the maximal element, and the constant functions, rules $0$ and $255$, are minimal elements. However, only rules $200$ and $236$ are above both $0$ and $255$, while rules  $4$, $12$, $68$ and $76$ are only above rule $0$, and rules $205$, $207$, $221$ and $223$ are only above rule $255$. The following result explains the reason behind this.  

\begin{lemma}
Let $0^\mathbb{Z}$ and $1^\mathbb{Z}$ be the all $0$'s and all $1$'s constant sequences. Let $\tau \in \End(A^{\mathbb{Z}})$ be an idempotent.
\begin{enumerate}
\item $0 \leq \tau$ if and only if $\tau(0^\mathbb{Z}) = 0^\mathbb{Z}$.
\item $255 \leq \tau$ if and only if $\tau(1^\mathbb{Z}) = 1^\mathbb{Z}$.
\end{enumerate}
\end{lemma}
\begin{proof}
We shall only prove point (1), as point (2) is done analogously. As rule $0$ is a constant map, we always have $0 \circ \tau = 0$. Now, for any $x \in A^\mathbb{Z}$, we have
\[ \tau \circ 0 (x) = \tau( 0^\mathbb{Z}). \]
Hence, $\tau( 0^\mathbb{Z}) = 0^\mathbb{Z}$ if and only if $\tau \circ 0 = 0$, which holds if and only if $0 \leq \tau$. 
\end{proof}

It is also worth noting that every automorphism $r$ of $\End(A^\mathbb{Z})$ preserve the natural order on idempotents: if $\sigma, \tau \in \End(A^{\mathbb{Z}})$ are idempotents such that $\sigma \leq \tau$, then $\sigma^r \leq \tau^r$. In particular, observe that $0^c = 255$, $0^\star = 0$ and $255^\star = 255$. Now, for example, since $4^c = 223$ and $0 \leq 4$, we must have $255 \leq 223$. Other relations may be deduced similarly.      

We say that a semigroup $S$ is \emph{ECA-maximal} if it satisfy two conditions:
 \begin{enumerate}
 \item All the elements of $S$ are ECA.
 \item If $S^\prime$ is a subsemigroup of $\End(A^\mathbb{Z})$ whose elements are all ECA and $S \subseteq S^\prime$, then $S=S^\prime$. 
 \end{enumerate}

Since the identity is an ECA, any ECA-maximal semigroup must be actually a monoid. By closure of the operation, an ECA-maximal monoid $S$ must satisfy that $X^2 \in S$ for all $X \in S$. Hence, in order to determine the ECA-maximal monoids, we consider the set 
\[ Q := \{ X :  X \circ X \text{ is a ECA } \}. \]
Direct computations show that 
 \[ Q = E \cup \{ 8, 51, 64, 239, 253 \}. \] 

\begin{proposition}\label{monoids}
Let $S$ be a ECA-maximal monoid. Then $S$ is equal to one of the following monoids:
\begin{itemize}
	\item $M_1 = \left\{ 0,4,8,12,64,68,200,204,255  \right\}$
	\item $M_2=\left\{  0,204,207,221,223,236,239,253,255  \right\}$
	\item $M_3=\left\{   0,12,204,221,255   \right\}$
	\item $M_4=\left\{   0,68,204,207,255  \right\}$
	\item $M_5=\left\{  0,51,204,255  \right\}$
	\item $M_6=\left\{   0,76,204,255  \right\}$
	\item $M_7=\left\{  0,204,205,255  \right\}$
\end{itemize}
\end{proposition}
\begin{proof}
For each rule $X$, let $T_X$ be the set of two-sided companions of $X$ in $Q$:
\[ T_X := \{ Y  : Y \in Q,  X \circ Y \text{ and } Y \circ X \text{ are both ECA } \}. \]
As a first step we obtain $T_X$ for all $X \in Q$. First, note that $T_{0} = T_{204} = T_{255} = Q$. Furthermore,
\begin{align*}
& T_4  = T_8 =  T_{64} =  T_{200}  = \{ 0, 4, 8, 12, 64, 68, 200, 204, 255 \} , \\
& T_{12} = T_4 \cup \{ 221 \}, \ \  T_{68} = T_4 \cup \{ 207 \}, \\
& T_{223}  =  T_{236} = T_{239} =  T_{253}    = \{0,  207, 221, 223, 236, 239, 204, 253, 255\}, \\ 
 & T_{207} = T_{223} \cup \{ 68\} ,  \ \ T_{221} = T_{223} \cup \{ 12\}, 
\end{align*}
and $T_{51} =\{0, 51, 204, 255\}$, $T_{76} =\{0,76, 204, 255\}$, $T_{205} =\{0, 204, 205, 255\}$. Direct calculations show that $T_4$, $T_{223}$, $T_{51}$, $T_{76}$, $T_{205}$ are closed under composition, so they correspond to the ECA-maximal monoids $M_1$, $M_2$, $M_5$, $M_6$, and $M_7$, respectively. Note that $T_{12} \cap T_{221} = \{ 0, 12, 204, 221, 255 \}$ and $T_{68} \cap T_{207} = \{ 0, 68, 204, 207, 255 \}$; these sets are also closed under composition, and they correspond to the ECA-maximal monoids $M_3$ and $M_4$, respectively.     
\end{proof}

Unlike the case of groups, there is a huge number of semigroups of small order: for example, it is known that, up to isomorphism and anti-isomorphism, there are 836,021 semigroups of order $7$, 1,843,120,128 semigroups of order $8$, and 52,989,400,714,478 semigroups of order $9$ (\cite[Table 4.2]{Distler} and \cite[p. 4]{Distler2}). Therefore, despite their small order, it is an interesting goal to study the structure of the monoids given by Proposition \ref{monoids}. 

For any subsemigroup $S$ of $\End(A^{\mathbb{Z}})$, and any automorphism $r$ of $\End(A^{\mathbb{Z}})$, we use the notation
\[ S^r = \{ s^r \in \End(A^{\mathbb{Z}}) : s \in S \}. \]
Direct calculations show that
\[ (M_1)^c = M_2, \quad (M_3)^\star = M_4, \quad (M_6)^c = M_7, \]
so 
\[ M_1 \cong M_2, \quad M_3 \cong M_4, \quad M_6 \cong M_7.  \] 
The monoid $M_5$ is stable under reflection and complement (i.e. $(M_5)^\star = M_5$ and $(M_5)^c = M_5$). Furthermore, we see that $M_5$ is not isomorphic to any of the other monoids given in Proposition \ref{monoids} since it is the only one with a nontrivial group of units (as $U(M_5) = \{ 51, 204 \}$). 

Tables \ref{CayleyM1} and \ref{CayleyM3} provide the Cayley tables of the monoids $M_1$, $M_3$, $M_5$, and $M_6$, which we take as representatives of their respective isomorphism classes. 

\begin{table}[h] \centering
\caption{Cayley table of the monoid $M_1$.} \label{CayleyM1}
\begin{tabular}{|c||c|c|c|c|c|c|c|c|c|}
\hline
$M_1$ & 0 & 4 & 8 & 12 & 64 & 68 & 200 & 204 &255 \\
\hline \hline
0 & 0 & 0 & 0 & 0 & 0 & 0 & 0 & 0 & 0 \\
\hline
4 & 0 & 4 & 8 & 12 & 64 & 68 & 0 & 4 & 0 \\
\hline
8 & 0 & 0 & 0 & 0 & 0 & 0 & 8 & 8 & 0 \\
\hline
12 & 0 & 4 & 8 & 12 & 64 & 68 & 8 & 12 &  0 \\
\hline
64 & 0 & 0 & 0 & 0 & 0 & 0 & 64 & 64 & 0 \\
\hline
68 & 0 & 4 & 8 & 12 & 64 & 68 & 64 & 68 & 0 \\
\hline
200 & 0 & 0 & 0 & 0 & 0 & 0 & 200 & 200 &255 \\
\hline
204 & 0 & 4 & 8 & 12 & 64 & 68 & 200 & 204 & 255 \\
\hline
255 & 255 & 255 & 255 & 255 & 255 & 255 & 255 & 255 & 255 \\
\hline
\end{tabular}
\end{table}

\begin{table}[h]\centering
\caption{Cayley tables of the monoids $M_3$, $M_5$, and $M_6$.}\label{CayleyM3}
 \begin{tabular}{|c||c|c|c|c|c|}
\hline 
$M_3$ & 0 & 12 & 204 & 221  & 255 \\
\hline\hline
0 & 0 & 0 & 0 & 0 & 0  \\
\hline
12 & 0 & 12 & 12 & 12 & 0 \\
\hline
204 & 0 & 12 & 204 & 221 & 255\\
\hline
221 & 255 & 221 & 221 & 221 & 255 \\
\hline
255 & 255 & 255 & 255 &255 & 255  \\
\hline 
\end{tabular} \quad 
    \begin{tabular}{|c||c|c|c|c|}
\hline
$M_5$ & 0 & 51 & 204 & 255 \\
\hline\hline
0 & 0 & 0 & 0 & 0   \\
\hline
51 & 255 & 204 & 51 & 0 \\
\hline
204 & 0 & 51 & 204 & 255 \\
\hline
255 & 255 & 255 & 255 & 255  \\
\hline
\end{tabular}  \quad  \quad  \begin{tabular}{|c||c|c|c|c|}
\hline
$M_6$ & 0 & 76 & 204 & 255 \\
\hline\hline
0 & 0 & 0 & 0 & 0    \\
\hline
76 & 0 & 76 &  76 & 0 \\
\hline
204 & 0 & 76 & 204 & 255 \\
\hline
255 & 255 & 255 & 255 & 255 \\
\hline
\end{tabular} 
\end{table}

With the exception of $M_5$, all these monoids are generated by idempotents. In fact, the monoids $M_3$ and $M_6$ are bands (all their elements are idempotents). Both $M_3$ and $M_6$ satisfy the identity
\[ X \circ Y \circ X = X \circ Y, \quad \forall X, Y  \in M_i, \ i \in \{ 3,6 \}, \] 
which imply that they are \emph{left regular bands}. These special kind of bands have been used in \cite{Brown} to define random walks and Markov chains. 

\begin{proposition}
\begin{enumerate}
\item For the monoid $M_1$, the Green's equivalence classes are:
\begin{itemize}
	\item $\mathcal{L}$-classes: $\{0 \}$, $\{ 4,12,68 \}$, $\{ 8\}$, $\{64\}$, $\{200\}$, $\{204\}$, $\{255\}$.
	\item $\mathcal{R}$-classes: $\{0,255\}$, $\{4\}$, $\{8\}$, $\{12\}$, $\{64\}$, $\{68\}$, $\{200\}$, $\{204\}$.
	\item $\mathcal{J}$-classes: $\{0,255\}$, $\{4,12,68\}$, $\{8\}$, $\{64\}$, $\{200\}$, $\{204\}$.
	\item $\mathcal{H}$-classes: $\{0\}$, $\{4\}$, $\{8\}$, $\{12\}$, $\{64\}$, $\{68\}$, $\{200\}$, $\{204\}$, $\{255\}$.
	\item $\mathcal{D}$-classes: $\{0,255\}$, $\{4,12,68\}$, $\{8\}$, $\{64\}$, $\{200\}$, $\{255\}$.
	
\end{itemize}

\item For the monoid $M_3$, the Green's equivalence classes are:
\begin{itemize}
	\item $\mathcal{L}$-classes: $\{0\}$, $\{12\}$, $\{221\}$, $\{204\}$, $\{255\}$.
	\item $\mathcal{R}$-classes: $\{0,255\}$, $\{12,221\}$, $\{204\}$.
	\item $\mathcal{J}$-classes: $\{0,255\}$, $\{12,221\}$, $\{204\}$.
	\item $\mathcal{H}$-classes: $\{0\}$, $\{12\}$, $\{221\}$, $\{204\}$, $\{255\}$.
	\item $\mathcal{D}$-classes: $\{0,255\}$, $\{12,221\}$, $\{204\}$.
	
\end{itemize}

\item For the monoid $M_5$, the Green's equivalence classes are:
\begin{itemize}
	\item $\mathcal{L}$-classes: $\{0\}$, $\{51,204\}$, $\{255\}$.
	\item $\mathcal{R}$-classes: $\{0,255\}$, $\{51,204\}$.
	\item $\mathcal{J}$-classes: $\{0,255\}$, $\{51\}$, $\{204\}$.
	\item $\mathcal{H}$-classes: $\{0\}$, $\{51,204\}$, $\{255\}$.
	\item $\mathcal{D}$-classes: $\{0,255\}$, $\{51,204\}$.
\end{itemize}

\item For the monoid $M_6$, the Green's equivalence classes are:
\begin{itemize}
	\item $\mathcal{L}$-classes: $\{0\}$, $\{76\}$, $\{204\}$, $\{255\}$.
	\item $\mathcal{R}$-classes: $\{0,255\}$, $\{76\}$, $\{204\}$.
	\item $\mathcal{J}$-classes: $\{0,255\}$, $\{76\}$, $\{204\}$.
	\item $\mathcal{H}$-classes: $\{0\}$, $\{76\}$, $\{204\}$, $\{255\}$.
	\item $\mathcal{D}$-classes: $\{0,255\}$, $\{76\}$, $\{204\}$.
\end{itemize}
\end{enumerate}
\end{proposition}

An important submonoid of $M_1$ is 
\[ M_1^\prime = M_1 - \{ 255 \}. \]
Note that $M_1^\prime$ has a zero element, which is rule $0$. Moreover, it is \emph{right reversible} in the sense that $M_1^\prime a \cap M_1^\prime b \neq \emptyset$ for all $a,b \in M_1^\prime$. These kind of semigroups were studied by Ore, who showed that any right reversible cancellative semigroup can be embedded in a group (\cite[Theorem 1.23]{CP}). The monoid $M_1$ itself is not right reversible. 

Another important property of $M_1^\prime$ is that it is $\mathcal{R}$-trivial. This means that Green's $\mathcal{R}$-relation is trivial, or, equivalently, that $aM_1^\prime = bM_1^\prime$ implies $a=b$. The $\mathcal{R}$-trivial monoids are an important class of finite monoids which have applications in the analysis of Markov chains (see \cite[Ch. 2]{Steinberg} and \cite{Markov}). This is a generalization of the work in \cite{Brown}, as left regular bands are precisely those $\mathcal{R}$-trivial monoids whose elements are all idempotents. 

\section{Conclusions} \label{sec-con}
 
The main conclusions of this chapter are the following:
\begin{enumerate}
\item It is possible to divide the 256 ECA into four classes, labeled by Class $N$ with $N=0,1,2,3$, using the numbers of left and right quasi-companions. Intuitively, these four classes are measuring the complexity of a given rule in an indirect way, by observing its behavior in the composition with other ECA. 

\item The 13 idempotent ECA may be ordered using the natural partial order on idempotents and the corresponding Hasse diagram is given by Fig. \ref{hasse_idem}.

\item Any semigroup of ECA must be contained in one of the seven monoids $M_i$, $i=1, \dots, 7$. Using the reflection and complement automorphisms, we see that $M_1 \cong M_2$, $M_3 \cong M_4$ and $M_6 \cong M_7$. Remarkably, $M_1$ contains a submonoid $M_1^\prime$ of order $8$ that has an element zero, is right reversible and its Green's $\mathcal{R}$-relation is trivial; these kind of monoids are important as they have application in the analysis of Markov chains.  
\end{enumerate}

Some possibilities for future work that expand the results of this chapter are the following:
\begin{enumerate}
\item Explore the behavior of the composition between two-dimensional cellular automata. 

\item Identify more one-dimensional cellular automata that are idempotent and compare them, using the natural partial order, with the idempotent ECA.  

\item Identify and analyze monoids whose elements are quasi-elementary cellular automata. 
 
\item Generalize the results obtained in this chapter for one-dimensional cellular automata over larger alphabets. 
\end{enumerate}


\textbf{Acknowledgement}: \\
We are very grateful with 	Andrew Adamatzky, Genaro Martinez, and Georgios Sirakoulis for their invitation to write this chapter. The first author was supported by a CONACYT Basic Science Grant (No. A1-S-8013) from the Government of Mexico. The second author was supported by a CONACYT Postgraduate National Scholarship.


\section*{Appendix} \label{appendix}
\addcontentsline{toc}{section}{Appendix}

In this appendix we include the tables related with the results of Section \ref{sec-comp}. As the number of left and right companions, factorizations, and left and right quasi-companions is preserved by the equivalences of ECA, our tables only include one representative per ECA class, according to Table \ref{equivalences}.

\begin{table}[!h]\centering
\caption{Number of companions in Wolfram's Class II}\label{comp-II}
\scriptsize
\begin{tabular}{|c|c|c|c||c|c|c|c|}
\hline
\textbf{Rule's} &\textbf{No. Left} &\textbf{No. Right} &\textbf{No. ECA} & \textbf{Rule's} &\textbf{No. Left} &\textbf{No. Right} &\textbf{No. ECA}  \\ 
\textbf{class} & \textbf{companions} & \textbf{companions} & \textbf{factorizations} & \textbf{class} & \textbf{companions} & \textbf{companions} & \textbf{factorizations} \\ \hline
1 &32 &4 &24 &  51 & 256 & 256 & 4\\
2 &20 &60 &50 &  56 &4 &4 &4  \\
3 &32 &28 &14  &57 &0 &0 &0  \\
4 &34 &28 &16 &58 &0 &0 &0 \\
5 &6 &0 &0 &62 &2 &0 &0 \\
6 &4 &4 &4 &72 &14 &4 &16 \\
7 &8 &0 &0 &73 &0 &0 &0 \\
9 &0 &0 &0 &74 &0 &0 &0 \\
10 &0 &12 &8 &76 &32 &4 &4 \\
11 &4 &4 &8 &77 &0 &0 &0 \\
12 &32 &60 &42 &78 &0 &0 &0  \\
13 &2 &0 &4 &94 &0 &0 &0 \\
14 &2 &12 &4 &104 &4 &0 &0 \\
15 &12 &12 &0 &108 &6 &0 &0 \\
19 &34 &28 &12 &130 &0 &0 &0 \\
23 &8 &0 &0 &132 &14 &0 &0 \\
24 &0 &28 &24 &134 &0 &0 &0 \\
25 &2 &4 &0 &138 &4 &4 &8 \\
26 &0 &0 &0 &140 &20 &12 &0 \\
27 &4 &4 &0 &142 &0 &0 &0 \\
28 &4 &4 &4 &152 &2 &4 &0 \\
29 &8 &0 &0 &154 &2 &0 &0 \\
33 &14 &0 &0 &156 &0 &0 &0 \\
34 &32 &60 &42 &162 &2 &0 &4  \\
35 &20 &12 &0 &164 &4 &0 &0 \\
36 &8 &4 &8 &170 &12 &12 &0  \\
37 &4 &0 &0  &172 &4 &4 &0\\
38 &8 &0 &0 &178 &0 &0 &0  \\
42 &2 &12 &4 &184 &8 &0 &0 \\
43 &0 &0 &0 &200 &34 &28 &12 \\
44 &8 &0 &0 & 204 & 256 & 256 & 5   \\
46 &8 &4 &40 &232 &8 &0 &0   \\
50 &32 &4 &4& & & &  \\ \hline
\end{tabular}
\end{table}

\begin{table}[!h]\centering
\caption{Number of companions in Wolfram's Classes I, III and IV}\label{comp-I}
\scriptsize
\begin{tabular}{|l|r|r|r|r|r|}
\hline
\textbf{Wolfram's Class} &\textbf{Rules} &\textbf{No. Left companions} &\textbf{No. Right companions} &\textbf{No. ECA factorizations} \\\hline
\multirow{7}{*}{Class I}  & 0 & 256 & 256 & 818 \\
& 8 &20 &60 &50 \\
& 32 &34 &28 &16 \\
&40 &4 &4 &4 \\
&128 &32 &4 &24 \\
&136 &32 &28 &14 \\
&160 &6 &0 &0 \\
&168 &8 &0 &0 \\ \hline
\multirow{11}{*}{Class III} &18 &14 &4 &16 \\
&22 &4 &0 &0 \\
&30 &2 &0 &0 \\
&45 &2 &0 &0 \\
&60 &12 &12 &6 \\
&90 &0 &0 &6 \\
&105 &0 &0 &0 \\
&122 &0 &0 &0 \\
&126 &0 &4 &8 \\
&146 &0 &0 &0 \\
&150 &0 &0 &0 \\
\hline
\multirow{4}{*}{Class IV} &41 &0 &0 &0 \\
&54 &6 &0 &0 \\
&106 &2 &0 &0 \\
&110 &2 &0 &0 \\
\hline
\end{tabular}
\end{table}

\begin{table}[!htp]\centering
\caption{Number of quasi-companions in Wolfram's Class II}\label{qcomp-II}
\scriptsize
\begin{tabular}{|c|c|c||c|c|c|}\hline
\textbf{Rule's} &\textbf{No. Left} &\textbf{No. Right} & \textbf{Rule's} &\textbf{No. Left} &\textbf{No. Right}  \\ 
\textbf{class} & \textbf{quasi-companions} & \textbf{quasi-companions} & \textbf{class} & \textbf{quasi- companions} & \textbf{quasi-companions}  \\ \hline
1 &32 &8 & 51 & 256 & 256\\
2 &28 &120 &56 &16 &8 \\
3 &46 &48 &57 &0 &0 \\
4 &38 &56 &58 &0 &0 \\
5 &6 &0 &62 &4 &0 \\
6 &4 &8 &72 &14 &8 \\
7 &20 &0 &73 &0 &0 \\
9 &4 &0 &74 &4 &0 \\
10 &12 &24 &76 &32 &8 \\
11 &28 &8 &77 &0 &0 \\
12 &52 &88 &78 &0 &0 \\
13 &14 &0 &94 &0 &0 \\
14 &30 &36 &104 &4 &0 \\
15 & 256 & 256 &108 &6 &0 \\
19 &34 &40 &130 &4 &0  \\
23 &8 &0 &132 &18 &0  \\
24 &8 &56 &134 &0 &0  \\
25 &8 &8 &138 &28 &8  \\
26 &4 &0 &140 &22 &16  \\
27 &8 &4 &142 &0 &0  \\
28 &16 &8 &152 &8 &8  \\
29 &28 &0 &154 &6 &0   \\
33 &18 &0 &156 &0 &0  \\
34 &52 &88 &162 &14 &0   \\
35 &22 &16 &164 &4 &0  \\
36 &8 &8 & 170 & 256 & 256 \\
37 &4 &0 &172 &8 &4  \\
38 &16 &0 &178 &0 &0  \\
42 &30 &36 &184 &28 &0  \\
43 &0 &0 &200 &34 &40   \\
44 &16 &0 & 204 & 256 & 256 \\
46 &16 &8 &232 &8 &0    \\
50 &32 &8   & & &  \\ \hline
\end{tabular}
\end{table}

\begin{table}[!htp]\centering
\caption{Number of quasi-companions in Wolfram's Class I, III and IV}\label{qcomp-I}
\scriptsize
\begin{tabular}{|l|r|r|r|r|r|}\hline
\textbf{Wolfram's class} &\textbf{Rule} &\textbf{No. Left quasi-companions} &\textbf{No. Right quasi-companions} 
\\\hline
\multirow{7}{*}{Clase I} & 0 & 256 & 256 \\
&8 &28 &120 \\
&32 &38 &56 \\
&40 &4 &8 \\
&128 &32 &8 \\
&136 &46 &48 \\
&160 &6 &0 \\
&168 &20 &0 \\
\hline
\multirow{11}{*}{Clase III} &18 &14 &8 \\
&22 &4 &0 \\
&30 &12 &0 \\
&45 &6 &0 \\
&60 &20 &20 \\
&90 &0 &0 \\
&105 &0 &0 \\
&122 &0 &0 \\
&126 &0 &8 \\
&146 &0 &0 \\
&150 &0 &0 \\
\hline
\multirow{4}{*}{Clase IV} &41 &0 &0 \\
&54 &6 &0 \\
&106 &12 &0 \\
&110 &4 &0 \\
\hline
\end{tabular}
\end{table}


\end{document}